\documentclass[11pt]{article}
\usepackage{fullpage}
\usepackage{graphicx}
\usepackage{amssymb}
\usepackage{amsmath,amsthm}
\usepackage{color}
\usepackage{xspace}
\usepackage{algorithm}
\usepackage[noend]{algorithmic}
\usepackage{enumitem}

\usepackage{ifpdf}
\ifpdf    
\usepackage{hyperref}
\else    
\usepackage[hypertex]{hyperref}
\fi

\usepackage[margin=0.9in]{geometry}

\newcommand{\remove}[1]{}

\def\compactify{\itemsep=0pt \topsep=0pt \partopsep=0pt \parsep=0pt}

\newtheorem{theorem}{Theorem}[section]

\newtheorem{inv}[theorem]{Invariant}

\newtheorem{lemma}[theorem]{Lemma}

\newtheorem{obs}[theorem]{Observation}

\providecommand{\card}[1]{\lvert#1\rvert}

\providecommand{\eqdef}{:=}

\newcommand{\N}{{\mathbb N}}

\newcommand{\comment}[1]{}
\newcommand{\junk}[1]{}

\newcommand{\dg}{d_{\textit{out}}}
\newcommand{\din}{d_{\textit{in}}}
\newcommand{\orient}[2]{#1 \rightarrow #2}
\newcommand{\rt}{s}

\newcommand{\ri}{{\sf Recursive-Insertion}}
\newcommand{\rd}{{\sf Recursive-Deletion}}
\newcommand{\eri}{{\sf Efficient-Recursive-Insertion}}
\newcommand{\erd}{{\sf Efficient-Recursive-Deletion}}

\title{Orienting Fully Dynamic Graphs\\ with Worst-Case Time Bounds}

\author{
Tsvi Kopelowitz\thanks{University of Michigan. This work was done in part while this author was a postdoctorate researcher at the Weizmann Institute of Science, and was supported in part by a US-Israel BSF grant \#2010418,
an Israel Science Foundation grant \#897/13,
and by the Citi Foundation.
Email: \texttt{kopelot@gmail.com}
}
\and
Robert Krauthgamer\thanks{Weizmann Institute of Science.
Work supported in part by a US-Israel BSF grant \#2010418,
an Israel Science Foundation grant \#897/13,
and by the Citi Foundation.
Email: \texttt{robert.krauthgamer@weizmann.ac.il}
}
\and
Ely Porat\thanks{Bar-Ilan University. Email: \texttt{porately@cs.biu.ac.il}}
\and
Shay Solomon\thanks{Weizmann Institute of Science.
Email: \texttt{shay.solomon@weizmann.ac.il}. This work is supported by the Koshland Center for basic Research.}
}

\begin{document}

\maketitle

\begin{abstract}

In edge orientations, the goal is usually to orient (direct) the edges of
an undirected $n$-vertex graph $G$ such that all out-degrees are bounded.
When the graph $G$ is fully dynamic, i.e., admits edge insertions and deletions,
we wish to maintain such an orientation while keeping a tab on
the update time.
Low out-degree orientations turned out to be a surprisingly useful tool,
with several algorithmic applications involving static or dynamic graphs.

Brodal and Fagerberg (1999) initiated the study of the edge orientation problem
in terms of the graph's arboricity, which is very natural in this context.
They provided a solution
with constant out-degree and \emph{amortized} logarithmic update time
for all graphs with constant arboricity,
which include all planar and excluded-minor graphs.
However, it remained an open question (first proposed by Brodal and Fagerberg, later by others)
to obtain similar bounds with worst-case update time.

We resolve this 15 year old question in the affirmative,
by providing a simple algorithm with worst-case bounds
that nearly match the previous amortized bounds.
Our algorithm is based on a new approach of a combinatorial invariant,
and achieves a logarithmic out-degree with logarithmic worst-case update times.
This result has applications in various dynamic graph problems
such as maintaining a maximal matching, where we obtain $O(\log n)$ worst-case update time compared to the
$O(\frac{\log n}{\log\log n})$ amortized update time of Neiman and Solomon (2013).
\end{abstract}

\section{Introduction}

Low out-degree \emph{orientations} are a very useful tool in designing algorithms.
The idea is to orient (direct) the edges of an undirected graph $G = (V,E)$
while providing a guaranteed upper bound on the out-degree of every vertex.
Formally, a {\em $c$-orientation} refers to
an orientation in which the out-degree of every vertex is at most $c$.
An exciting example of the power of graph orientations can be seen
in the seminal ``color-coding paper'' \cite{AYZ95}, where orientations are used
to develop more efficient algorithms for finding simple cycles and paths.
Another fundamental example can be seen in maintaining data structures for answering
adjacency queries~\cite{CE91,BF99,Kowalik07},
where a $c$-orientation for $G$ is used to quickly answer
adjacency queries on $G$ in $O(c)$ time using only linear space,
and these ideas were further generalized to answer short-path queries~\cite{KK06}.
Additional examples where low-degree orientations were exploited algorithmically
include load balancing~\cite{CSW07}, maximal matchings~\cite{NS13}, counting subgraphs in sparse graphs~\cite{DT13}, prize-collecting TSPs and Steiner Trees~\cite{EKM12}, reporting all maximal independent sets~\cite{Eppstein09}, answering dominance queries~\cite{Eppstein09}, subgraph listing problems (listing triangles and $4$-cliques) in planar graphs \cite{CE91}, computing the girth~\cite{KK06}, and more.

In many contexts, algorithmic efficiency can be improved significantly
if for each edge, one of its endpoints becomes
``responsible'' for data transfer taking place on that edge.
Such a responsibility assignment can be naturally achieved by orienting graph's edges and letting each vertex be responsible only for its outgoing edges.
For example, when we ask a vertex $u$ to compute some function of dynamic data residing locally at $u$'s neighbors, we would like to avoid scanning all of $u$'s neighbors. Given a $c$-orientation, whenever the local data in $u$ changes,
$u$ is responsible to update all its outgoing neighbors (neighbors of $u$ through edges that are oriented out of $u$).
In contrast, $u$ need not update any of its (possibly many) incoming neighbors (neighbors of $u$ through edges that are oriented into $u$) about this change.
When $u$ wishes to compute the function, it only needs to scan its outgoing neighbors in order to get the full up-to-date data.
Such responsibility assignment is particularly useful in dynamic networks (see \cite{NS13} for an example),
but is highly applicable also in other contexts,
such as reducing the message complexity of distributed or self-stabilizing networks,
or reducing local memory constraints in such systems, e.g., a router would only store information about its $c$ outgoing neighbors.

\paragraph{Dynamic Graphs.}
Our focus here is on maintaining low out-degree orientations of fully dynamic graphs on $n$ fixed vertices, where edge updates (insertions and deletions) take place over time. The goal is to develop efficient and simple algorithms that guarantee that the maximum out-degree in the (dynamic) orientation of the graph is small. In particular, we are interested in obtaining non-trivial update times
that hold (1) in the \emph{worst-case}, and (2) \emph{deterministically}.
Notice that in order for an update algorithm to be efficient, the number of
\emph{edge reorientation} (done when performing an edge update) must be small,
as this number is an immediate lower bound for the algorithm's update time.

The maximum out-degree achieved by our algorithms will be expressed in terms
of the sparsity of the graph, as measured by the \emph{arboricity} of $G$
(defined below),
which is (as we shall see) a natural lower bound for the maximum out-degree of any orientation.

\paragraph{Arboricity.}
The \emph{arboricity}
of an undirected graph $G=(V,E)$ is defined as
$\alpha(G) =  \max_{U\subseteq V} \left\lceil \frac{\card{E(U)}}{\card{U}-1} \right\rceil$,
where $E(U)$ is the set of edges induced by $U$
(which we assume has size $\card{U}\ge 2$).
This is a concrete formalism for the notion of everywhere-sparse graphs; for a graph with arboricity at most $\alpha$, every subgraph of this graph has arboricity at most $\alpha$ as well.
The notion of arboricity, as well as other sparseness measurements such as \emph{thickness, degeneracy} or \emph{density} (which are all equal up to constant factors) have been subject to extensive research.
Most notable in this context is the family of graphs with constant arboricity, which includes all excluded-minor graphs, and in particular planar graphs and bounded-treewidth graphs.

A key property of bounded arboricity graphs which has been exploited in various algorithmic applications is the following Nash-Williams Theorem.

\begin{theorem}[Nash-Williams~\cite{NASH61,Nash64}]
A graph $G=(V,E)$ has arboricity $\alpha(G)$ if and only if $\alpha(G)>0$ is the smallest number of sets $E_1,\ldots,E_{\alpha(G)}$ that $E$ can be partitioned into, such that each subgraph $(V,E_{i})$ is a forest.
\end{theorem}

The Nash-Williams Theorem implies that one can orient the edges of an undirected graph $G=(V,E)$ with bounded arboricity $\alpha(G)$ such that the out-degree of each vertex is at most $\alpha(G)$. To see this, consider the partition $E_1,\ldots, E_{\alpha(G)}$ guaranteed by the Nash-Williams Theorem. For each forest $(V,E_i)$ and for each tree in that forest, designate one arbitrary vertex as the root of that tree, and orient all edges towards that root. In each oriented forest the out-degree of each vertex is at most 1, hence in the union of the oriented forests the out-degree of each vertex is at most $\alpha(G)$.
There exists
a polynomial-time algorithm for computing the exact arboricity $\alpha(G)$~\cite{GW92},
and a linear-time algorithm for computing a $(2\alpha(G) -1)$-orientation for a static graph $G$~\cite{AMZ97}.

For every graph $G$, the maximum out-degree (of its edge orientations)
is closely related to $\alpha(G)$:
There exists a \emph{static} orientation of maximum out-degree at most
$\alpha(G)$ (by the above argument using the Nash-Williams Theorem),
but the maximum out-degree is also easily seen to be at least $\alpha(G)-1$
(for every orientation).%
\footnote{To see this, let $U\subset V$ be such that
$\left\lceil \frac{|E(U)|}{|U|-1} \right\rceil = \alpha(G)$,
hence $\frac{|E(U)|}{|U|-1} > \alpha(G)-1$.
For every orientation, the maximum out-degree in $G$
is at least the average out-degree of vertices in $U$,
which in turn is at least
$\frac{\card{E(U)}}{\card{U}} > \frac{\card{U}-1}{\card{U}} (\alpha(G)-1)$.
The bound now follows from both $\alpha(G)$ and the maximum out-degree being integers.
}
In other words, the arboricity is a very natural candidate as a measure
of sparsity in the context of low out-degree orientations.

\subsection{Main Result}
\label{sec:results}

We obtain efficient algorithms for maintaining a low out-degree orientation of a fully dynamic graph $G$ with arboricity bounded by $\alpha$,
such that the out-degree of each vertex is small and the running time of all update operations is bounded in the worst-case.
Specifically, our algorithms maintain
\begin{itemize} \compactify
\item a maximum out-degree
$\Delta  \leq \inf_{\beta>1} \{\beta\cdot \alpha(G) + \lceil \log_\beta n \rceil\}$, and
\item insertion and deletion update times $O(\beta\cdot \alpha\cdot \log n)$
and $O(\Delta )$, respectively.
\end{itemize}
 Our algorithms have the following additional features:
(1) they are deterministic, (2) they are simple, and (3) each edge update changes the orientation of at most $\Delta +1$ edges.

Notice that for constant $\alpha$, we can take $\beta=2$ and all of our bounds translate to $O(\log n)$.
In other words, for any graph of constant arboricity, we can maintain an $O(\log n)$-orientation with $O(\log n)$ worst-case update time.
Previous work, which is discussed next, only obtained efficient \emph{amortized} update time bounds, in contrast to our bounds which are all in the worst-case. Our results resolve a fundamental open question raised by Brodal and Fagerberg~\cite{BF99} and restated by Erickson~\cite{Erickson06}, of whether such bounds are possible.

\subsection{Comparison with Previous Work}\label{subsec:related_work}
The dynamic setting in our context was pioneered by Brodal and Fagerberg~\cite{BF99} and extended by Kowalik~\cite{Kowalik07}.
Brodal and Fagerberg~\cite{BF99} showed that it is possible to maintain a $4\alpha$-orientation of a fully dynamic graph $G$ whose arboricity is always at most $\alpha$. They proved that their algorithm is $O(1)$-competitive against the number of re-orientations made by any algorithm, regardless of that algorithm's actual running time. They then provided a specific strategy for re-orienting edges which showed that the, for $\alpha=O(1)$, the insertion update time of their algorithm is {\em amortized} $O({1})$ while the cost of deletion is {\em amortized} $O(\log n)$ time. Kowalik~\cite{Kowalik07} provided another analysis of Brodal and Fagerberg's algorithm, showing that the insertion update time is {\em amortized} $O({\log n})$ while the cost of deletion is {\em worst-case} $O(1)$ time. Kowalik further showed that it is possible to support insertions in {\em amortized} $O(1)$ time and deletions in worst-case $O(1)$ time
by using an $O(\log n)$-orientation. These algorithms have been used as black-box components in many applications for dynamic graphs.

Algorithms with amortized runtime bounds may be insufficient for many real-time applications where infrequent costly operations might cause congestion in the system at critical times. Exploring the boundaries between amortized and worst-case bounds is also important from a theoretical point of view,
and has received a lot of research attention. The algorithms of Brodal and Fagerberg~\cite{BF99} and Kowalik~\cite{Kowalik07} both incur a \emph{linear worst-case update time}, on which we show an exponential improvement. As mentioned above, this answers an open question raised by Brodal and Fagerberg~\cite{BF99} (and restated by Erickson in~\cite{Erickson06})
of whether such bounds are obtainable.

\subsection{Our Techniques}\label{subsec:techniques}

The algorithm of Brodal and Fagerberg~\cite{BF99} is very elegant,
but it is not clear if it can be deamortized as it is inherently amortized.
The key technical idea we introduce is to maintain a combinatorial invariant,
which is very simple in its basic form:
for every vertex $u\in V$, at least (roughly) $\alpha$
outgoing edges are directed towards vertices
with almost as large out-degree, namely at least $\dg(u)-1$ (where $\dg(u)$ is the out-degree of $u$). Such edges are called \emph{valid} edges.
We prove in Section~\ref{sec:invariant} that this combinatorial invariant
immediately implies the claimed upper bound on $\Delta$.

An overview of the algorithms that we use for, say, insertion, is as follows.
When a new edge $(u,v)$ is added, we first orient it, say, from $u$ to $v$ guaranteeing that the edge is valid.
We now check if the invariant holds,
but the only culprit is $u$, whose out-degree has increased.
If we know which of the edges leaving $u$ are the ``special'' valid edges needed to maintain the invariant, we scan them to see if any of them are no longer valid (as a result of the insertion), and if there is such an edge we \emph{flip} its orientation,
and continue recursively with the other endpoint of the flipped edge.
This process indeed works, but it causes difficulty during an edge deletion
--- when one of the $\alpha$ special valid edges leaving $u$ is deleted,
a replacement may not even exist.

Here, our expedition splits into two different parts.
We first show an extremely simple (but less efficient) algorithm
that maintains a stronger invariant in which for every vertex $u\in V$,
\emph{all} of its out-going edges are valid.
This approach immediately gives the claimed upper bound on $\Delta$,
with update time roughly $O(\log ^2 n)$ for sparse graphs.

In the second part we refine the invariant using another idea of
\emph{spectrum-validity}, which roughly speaking uses the following invariant:
for every vertex $u\in V$ and for every $1\leq i \leq \frac{\deg(u)}{\alpha}$, at least $i\cdot \alpha$ of its outgoing edges are directed towards vertices with degree at least $\dg(u)-i$.
This invariant is stronger than the first invariant (which seemed algorithmically challenging) and weaker than the second invariant (whose bounds were less efficient than desired as it needed to guarantee validness for all edges).
Furthermore, maintaining this invariant is more involved algorithmically,
and one interesting aspect of our algorithm is that during an insertion process,
it does not scan the roughly $\alpha$ neighbors with degree at least $\dg(u)-1$, as one would expect,
but rather some other neighbors picked in a careful manner.
Ultimately, this methodology yields the improved time bounds claim in Section
\ref{sec:results}.

\subsection{Selected Applications}\label{subsec:applications}

We only mention two applications here by stating their theorems for graphs with arboricity bounded by a constant. We discuss these applications and some other ones with more detail in Appendix~\ref{app:applications}.

\begin{theorem}[Maximal matching in fully dynamic graphs]
Let $G=(V,E)$ be an undirected fully dynamic graph with arboricity bounded by a constant. Then one can deterministically maintain a maximal matching of $G$ such that the worst-case time per edge update
is $O(\log n)$.
\end{theorem}

\begin{theorem}[Adjacency queries in fully dynamic graphs]
Let $G=(V,E)$ be an undirected fully dynamic graph with arboricity bounded by a constant. Then one can deterministically answer adjacency queries on $G$ in $O(\log \log \log n )$ worst-case time where the deterministic worst-case time per edge update
is $O(\log n \cdot \log \log \log n)$.
\end{theorem}

\subsection{Preliminaries}
\label{sec:prelims}

Let $G=(V,E)$ be an undirected graph, and denote $n=|V|$.
The \emph{arboricity} of the graph, denoted $\alpha(G)$,
is the smallest integer $\alpha\ge1$ such that
all nonempty $U \subseteq V$
satisfy $\card{E(U)}\leq \alpha(\card{U}-1)$, where $E(U)=\{(u,v)\in E:\ u,v\in U\}$.

An \emph{orientation} of the undirected edges of $G$ assigns a direction to every edge $e\in E$, thereby turning $G$ into a digraph. We will use the notation $\orient u v$ to indicate that the edge $e=(u,v)$ is oriented from $u$ to $v$.
Given such an orientation,
let $N^+(u)\eqdef \{v\in V:\ \orient u v\}$ denote the set of \emph{outgoing neighbors} of $u$, i.e., the vertices connected to $u$ via an edge leaving it,
and let $\dg(u) \eqdef |N^+(u)|$ denote the number of \emph{outgoing edges} of $u$ in this orientation, i.e., the \emph{out-degree} of $u$.
Similarly, let $N^-(u)\eqdef \{v\in V:\ \orient v u\}$ denote the set of \emph{incoming neighbors} of $u$, and let $\din(u) \eqdef |N^-(u)|$.
Finally, we denote by
$ \Delta \eqdef \max_{v\in V}\ \dg(v) $
the maximum out-degree of a vertex in the graph (under the given orientation).

Our algorithms will make use of the following heap-like data structure.

\begin {lemma} \label{lem:structure}
Let $X$ be a dynamic set,
where each element $x_i\in X$ has a key $k_i\in \N$ that may change with time,
and designate a fixed element $x_0\in X$ to be the \emph{center} of $X$
(although its key $k_0$ may change with time).
Then there is a data structure that maintains $X$ using $O(\card{X}+k_0)$ words of space,
and supports the following operations with $O(1)$ worst-case time bound
(unless specified otherwise):
\begin{itemize}[noitemsep,nolistsep]
\item {\sf ReportMax}$(X)$: return a pointer to an element from $X$ that has the maximum key.
\item {\sf Increment}$(X,x)$: given a pointer to an element $x\in X \setminus \{x_0\}$, increment the key of $x$.
\item {\sf Decrement}$(X,x)$: given a pointer to an element $x\in X  \setminus \{x_0\}$, decrement the key of $x$.
\item {\sf Insert}$(X,x_i,k_i)$: insert a new element $x_i$ with key $k_i\le k_{0} +1$ into $X$.
\item {\sf Delete}$(X,x)$: given a pointer to an element $x\in X  \setminus \{x_0\}$, remove $x$ from $X$.
\item {\sf IncrementCenter}$(X)$: increment $k_0$. This operation takes $O(k_{0})$ worst-case time.
\item {\sf DecrementCenter}$(X)$: decrement $k_{0}$ (unless $k_{0} = 1$). This operation takes $O(k_{0})$ worst-case time.
\end{itemize}
\end{lemma}

\begin{proof}

For each distinct key value $k$ such that there is some element $x\in X$ with that key $k$ we maintain a list $L_k$. The lists $L_k$ for $k\geq k_0+2$ are maintained in a sorted list $L$ in a natural way. We also maintain pointers to the head and tail of $L$. Notice that the tail gives direct access to an element with the largest key.
The lists $L_k$ for $k\leq k_{0} +1$ are not maintained in the sorted list. Instead, we maintain an array of pointers $A$ of size $k_{0}+1$ such that the pointer at index $i$ in $A$ points to $L_i$.
All the operations except for {\sf IncrementCenter}$(X)$ and {\sf DecrementCenter}$(X)$ are implemented by a constant number of straightforward operations on these lists.
The operations {\sf IncrementCenter}$(X)$ and {\sf DecrementCenter}$(X)$ are implemented naively by rebuilding the array $A$ from scratch. Notice that the lists that are indexed by $A$ are unchanged during these operations.
\end{proof}

For each vertex $w \in V$, consider the (dynamic) set $X_w$ that contains $w$ and all its incoming neighbors,
where the key of each element in $X$ is given by its out-degree. The center element of $X_w$ will be $w$ itself.
Each vertex $w$ will have its own data structure (using Lemma~\ref{lem:structure}) for maintaining $X_w$.
In what follows, we denote this data structure by $H_w$,
and use it to find an incoming neighbor of $w$ with out-degree
at least $\dg(w) + 2$ (if one exists) in $O(1)$ time.

\begin{lemma}
The total space used to store the data structures $H_w$ for all $w \in V$ is $O(n+m)$ words,
where $m$ stands for the number of edges in the (current) graph.
\end{lemma}
\begin{proof}
By Lemma~\ref{lem:structure}, for each $w \in V$ the space usage is at most $O(1+\din(w) + \dg(w))$.
Summing over all vertices $w \in V$, the total space is $\sum_{w \in V} O(1+\din(w) + \dg(w)) = O(n+m)$.
\end{proof}

\section{Invariants for Bounding the Largest Out-degree}
\label{sec:invariant}

We assume throughout that the dynamic graph $G$ has, at all times,
arboricity $\alpha(G)$ bounded by some parameter $\alpha$, i.e., $\alpha(G)\leq \alpha$.
Let $\beta > 1$ be a parameter that may possibly depend on $n$ and $\alpha$
(it will be chosen later to optimize our bounds),
and define $\gamma \eqdef \beta \cdot \alpha$.

An edge $(u,v)\in E$ oriented such that $\orient u v $ is called \emph{valid} if $\dg(u) \leq \dg(v)+1$, and called \emph{violated} otherwise.
The following condition provides control (upper bound) on $\Delta$,
as proved in Theorem \ref{lem:invar}.
We refer to it as an \emph{invariant}, because we shall maintain
the orientation so that the condition is satisfied at all times.
\begin{inv} \label{weakest}
For each vertex $w$, at least $\min\{\dg(w),\gamma\}$ outgoing edges of $w$ are valid.
\end{inv}

\begin{theorem} \label{lem:invar}
If Invariant~\ref{weakest} holds,
then $\Delta \le \beta\cdot\alpha(G)  + \lceil \log_\beta n \rceil$.
\end{theorem}

\noindent{\bf Remark:}
We can pick $\beta=\beta(\alpha(G))$ that minimizes the upper bound on $\Delta$.
In particular, for $\alpha (G)\geq \log n$,
setting $\beta = 2$ yields $\Delta \leq O(\alpha(G))$;
for small $\alpha(G)$, say all $\alpha(G)\leq \sqrt{\log n}$,
setting $\beta = \sqrt[4]{\log n}$
yields $\Delta \leq O(\frac{\log n}{\log \log n})$.

\begin{proof}
Assume Invariant~\ref{weakest} holds, and suppose for contradiction there
is a ``source'' vertex $\rt\in V$ satisfying $\dg(\rt) > \gamma + \lceil \log_\beta n \rceil$.
Now consider the set $V_i$ of vertices reachable from $\rt$ by directed paths of length at most $i$ that use only valid edges.
Observe that for every $1 \le i \le \lceil \log_\beta n \rceil$
and every vertex $w \in V_i$,
 $$\dg(w)
 \ge \dg(\rt) - i
 > \gamma + \lceil \log_\beta n \rceil  - i
 \ge \gamma,$$
implying that at least $\gamma$ outgoing edges of  $w$ are valid.

We next prove by induction on $i$ that $|V_i| > \beta^i$
for all $1 \le i \le \lceil \log_\beta n \rceil$. For the base case $i = 1$, notice that $\rt$ has at least $\gamma$ valid outgoing edges and all of the corresponding outgoing neighbors of $\rt$ belong to $V_1$. Furthermore, $\rt$ belongs to $V_1$ as well. Thus $|V_1| \ge \gamma + 1 > \gamma \ge \beta$.
For the inductive step, suppose $|V_{i-1}| > \beta^{i-1}$;
observe that the total number of valid outgoing edges from vertices in $V_{i-1}$ is at least $\gamma |V_{i-1}|$,
and furthermore all these edges are incident only to vertices in $V_{i}$.
Since the graph's arboricity is $\alpha(G) \leq \alpha$,
we can bound $|V_{i}|-1 \geq \gamma |V_{i-1}| / \alpha(G) \geq \beta |V_{i-1}| > \beta^i$,
as claimed.

We conclude that $|V_{\lceil \log_\beta n \rceil}| > \beta^{\lceil \log_\beta n \rceil} \ge n$,
yielding a contradiction.
\end{proof}

Invariant~\ref{weakest} provides a relatively weak guarantee as
if $\dg(w) > \gamma$,
then we know only that $\gamma$ outgoing edges of $w$ are valid,
and have no guarantee on the out-degree of the other $\dg(w) - \gamma$
outgoing neighbors of $w$.
Consequently, it is nontrivial to maintain Invariant~\ref{weakest} efficiently,
and in particular, if one of the $\gamma$ valid edges (outgoing from $w$) is deleted,
the invariant might become violated, and it is unclear how to restore it efficiently.
We thus need another invariant,
namely, a stronger condition (so that Theorem~\ref{lem:invar} still applies)
that is also easy to maintain.
The next invariant is a natural candidate,
as it is simple to maintain (with reasonable efficiency).
\begin{inv}\label{inv:degree_inequality}
All edges in $G$ are valid.
\end{inv}

We first present in Section~\ref{sec:algo} a very simple algorithm
that maintains Invariant~\ref{inv:degree_inequality}
with update times $O(\Delta^2)$ and $O(\Delta)$ for insertion and deletion (of an edge), respectively.
This algorithm provides a strong basis for a more sophisticated algorithm,
developed in Section~\ref{sec:efficient_algo},
which maintains an intermediate invariant (stronger
than Invariant~\ref{weakest} but weaker than Invariant~\ref{inv:degree_inequality})
with update times $O(\gamma \cdot \Delta)$ and $O(\Delta)$ for
insertion and deletion, respectively.

\section{Worst-case Algorithm}\label{sec:algo}
We consider an infinite sequence of graphs $G_0,G_1,\ldots$ on a fixed vertex set $V$, where each graph $G_i=(V,E_i)$ is obtained from the previous graph $G_{i-1}$ by either adding or deleting a single edge. For simplicity, we assume that $G_0$ has no edges. Denote by $\alpha_i=\alpha(G_i)$ the arboricity of
$G_i$. We will maintain Invariant~\ref{inv:degree_inequality} while edges are inserted and deleted into and from the graph, which by Theorem~\ref{lem:invar} implies that the maximum out-degree $\Delta_i$ in the orientation of $G_i$ is bounded by $O(\inf_{\beta>1} \{\beta\cdot \alpha_i + \log_\beta n\})$.

For the rest of this section we fix $i$ and consider a graph $G_i$
that is obtained from a graph $G_{i-1}$ satisfying Invariant~\ref{inv:degree_inequality}
by either adding or deleting edge $e=(u,v)$.

\subsection{Insertions}

Suppose that edge $(u,v)$ is added to $G_{i-1}$ thereby obtaining $G_i$. We begin by orienting the edge from the endpoint with lower out-degree to the endpoint with larger out-degree (breaking a tie in an arbitrary manner).
So without loss of generality we now have $\orient u v$. Notice that the only edges that may be violated now are edges outgoing from $u$, as $\dg(u)$ is the only out-degree that has been incremented.
Furthermore, if some edge $\orient u v'$ is violated now,
then removing this edge will guarantee that there are no violated edges.
However, the resulting graph would be missing the edge $(u,v')$ just removed.
So we recursively insert the edge $(u,v')$, but orient it in the opposite direction (i.e., $\orient{v'} u$).
This means that we have actually \emph{flipped} the orientation of $(u,v')$, reverting $\dg(u)$ to its value before the entire insertion process took place.
This recursive process continues until all edges of the graph are valid.
Moreover, at any given time there is at most one ``missing'' edge,
and the graph obtained at the end of the process has no missing edges.
Our choice to remove a violated edge outgoing from $u$ (if such an edge exists)
guarantees that the number of recursive steps is at most $\Delta$,
as we will show later.
This insertion process is described in Algorithm~\ref{alg1}.

\begin{algorithm}
\caption{{\sf \ri$(G,(u,v))$}}
/* Assume without loss of generality that $\dg(u) \leq \dg(v) $ */
\begin{algorithmic}[1]
\label{alg1}
\STATE add $(u,v)$ to $G$ with orientation $\orient u v$
\STATE {\sf Insert}$(H_{v},u,\dg(u)-1)$ /* this key will be incremented in line 10 if needed */
\FOR {$v'\in N^+(u)$}
\IF {$\dg(u) > \dg(v')+1$ }
\STATE remove $(u,v')$ from $G$ /* now edge $(u,v')$ is missing */
\STATE {\sf Delete}$(H_{v'}, u)$
\STATE \ri$(G,(v',u))$ /* recursively insert $(u,v')$, but oriented $\orient {v'} u$ */
\STATE return
\ENDIF
\ENDFOR
\FOR {$v'\in N^+(u)$}
\STATE {\sf Increment}$(H_{v'},u)$ \ENDFOR
\STATE {\sf IncrementCenter}$(H_{u})$

\end{algorithmic}
\end{algorithm}

We remark that although in line 1 the out-degree of $u$ is incremented by 1, we do not update the new key of $u$
in the appropriate structures (i.e.,  $H_u$ and $H_{v'}$ for all $v' \in N^+(u)$),
because if the condition in line 4 succeeds
for some $v' \in N^+(u)$, the out-degree of $u$ will return to its original value, and we want to save the cost of incrementing and then decrementing the key for $u$ in all structures. However, if that condition fails for all $v'$,
we will perform the update in lines 9--11.

\paragraph{Correctness and Runtime Analysis}

For the following, assume that after the insertion of $(u,v)$ there is a violated edge $(u,v')$ (choosing one arbitrarily if it is not unique).

\begin{obs}\label{obs:violated_edge_insertion_bound}
Consider a call to \ri$(G,(u,v))$, with $G$'s orientation
satisfying Invariant~\ref{inv:degree_inequality},
and suppose that the condition in line 4 succeeds on $(u,v')$.
Then it must be that at the time the condition is tested, $\dg(u)=\dg(v')+2$.
\end{obs}

\begin{lemma}\label{lem:valid_graph_insertion_recursive_call}
During the execution of \ri$(G,(u,v))$ on a graph $G$
whose orientation satisfies Invariant~\ref{inv:degree_inequality},
every recursive call made to \ri \
has an input graph with an orientation that satisfies Invariant~\ref{inv:degree_inequality}.
\end{lemma}

\begin{proof}
The proof is by induction on the number of recursive steps of the algorithm. When edge $(u,v)$ is inserted and oriented, any violated edge must be outgoing from $u$. The algorithm checks (by brute-force) all such edges, and once it finds such an edge, that edge is removed thereby reverting $\dg(u)$ to its original value prior to the insertion of $(u,v)$. All edges in the remaining graph other than $(u,v)$ are valid by induction, and edge $(u,v)$ is valid as prior to the recursive call we had $\dg(u) \leq \dg(v)$.
\end{proof}

As a direct consequence from Lemma~\ref{lem:valid_graph_insertion_recursive_call} we have the following.

\begin{lemma}\label{cor:insertion_correct}
At the end of the execution of \ri \  on an input graph which has an orientation satisfying Invariant~\ref{inv:degree_inequality}, Invariant~\ref{inv:degree_inequality} holds for the resulting graph and orientation.
\end{lemma}

Next, we bound the total number of edge re-orientations performed and the time spent during the insertion process.

\begin{lemma}\label{lem:insertion_main}
The total number of recursive calls (and hence re-orientations) of \ri \ due to an insertion into $G$ is at most $\Delta +1$, and the total runtime is bounded by $O({\Delta }^2)$.
\end{lemma}

\begin{proof}
Consider an execution of \ri$(G,(u,v))$ and let $x$ be the out-degree of $u$ at the beginning of this execution. Furthermore, assume that during this execution we reach line 7 to call \ri$(G',(v',u))$ (notice that the difference between $G$ and $G'$ is two edges), and let $y$ be the out-degree of $v'$ at the start of this recursive call. Then due to Observation~\ref{obs:violated_edge_insertion_bound} it must be that $y=x-1$. In other words, in each consecutive recursive call the out-degree of the first vertex variable is decremented, and so the maximum number of recursive steps is bounded by $\Delta +1$.

The running time of each recursive step is dominated by a scan of the outgoing edges of some vertex which takes at most
$O(\Delta)$ time, yielding a total of $O(\Delta^2)$ time for the entire process. Finally, lines 9-11 are executed only once and take a total of $O(\Delta)$ time.
\end{proof}

\subsection{Deletions}\label{subsec:deletions}

Suppose that edge $(u,v)$ is deleted from $G_{i-1}$ thereby obtaining $G_i$.
Assume without loss of generality that in the orientation of $G_{i-1}$ we had $\orient u v$.
We begin by removing $(u,v)$ from our data structure.
Notice that the only edges that may be violated now are edges incoming into $u$.
Furthermore, if there is an edge $\orient{v'} u$ that is violated now,
then adding to the graph another copy of $(u,v')$ (producing a multi-graph)
that is oriented in the opposite direction (i.e., $\orient u v'$)
will guarantee that there are no violated edges.
However, the resulting multi-graph has an extra edge that should be deleted.
So we now recursively delete the original copy of edge $(u,v')$
(not the copy that was just added, oriented $\orient u v'$, which we keep).
This means that we have actually \emph{flipped} the orientation of $(u,v')$, reverting $\dg(u)$ to its value before the entire deletion process took place.
This recursive process will continue until all edges of the graph are valid.
Moreover, there is at most one duplicated edge at any given time, and the graph obtained at the end of the process has no duplicated edges.
Our choice to add a copy of a violated edge incoming to $u$ (if such an edge exists)
guarantees that the number of recursive steps is at most $\Delta$,
as we will show later.
This deletion process is described in Algorithm~\ref{alg2}.

\begin{algorithm}
\caption{\rd$(G,(u,v))$}
/* Assume without loss of generality that $(u,v)$ is oriented as $\orient u v$.
If another copy of $(u,v)$ is oriented $\orient v u$, assume without loss of generality that $\dg(u) \geq \dg(v) $ */
\begin{algorithmic}[1]
\label{alg2}
\STATE remove $(u,v)$ from $G$ /* if there are two copies of $(u,v)$, delete the one oriented as $\orient u v$ */
\STATE {\sf Delete}$(H_{v},u)$.
\STATE $v' \leftarrow$ ReportMax$(H_{u})$
\IF {$\dg (v') > \dg(u)+1$}
\STATE add $(u,v')$ to $G$ with orientation $\orient u {v'}$ /* now there are two copies of $(u,v')$ */
\STATE {\sf Insert}$(H_{v'}, u, \dg(u))$
\STATE \rd$(G,(v',u))$ /* recursively delete the copy of $(v',u)$ oriented as $\orient {v'} u$ */
\STATE return
\ENDIF
\FOR {$v'\in N^+(u)$}
\STATE {\sf Decrement}$(H_{v'},u)$
\ENDFOR
\STATE {\sf DecrementCenter}$(H_{u})$

\end{algorithmic}
\end{algorithm}

We briefly mention that although in line 1 we decrement the out-degree of $u$, we do not decrement the key of $u$ in the appropriate structures of Lemma~\ref{lem:structure}. This is because if we ever pass the condition in line 4, the out-degree of $u$ will return to its original value, and we want to save the cost of decrementing and then incrementing the key for $u$ in all structures. However, if the condition does not pass, then we will perform the update in lines 9--11.

\paragraph{Correctness and Runtime Analysis}
For the following, assume that after the deletion of $(u,v)$ there is a violated edge $(u,v')$ incoming into u (choosing one arbitrarily if it is not unique).

\begin{lemma}\label{lem:violated_edge_deletion_bound}
Consider a call to \ri$(G,(u,v))$, with $G$'s orientation
satisfying Invariant~\ref{inv:degree_inequality},
and suppose that the condition in line 4 succeeds on $(u,v')$.
Then it must be that at the time the condition is tested, $\dg(v')=\dg(u)+2$.
\end{lemma}

\begin{proof}
If $\dg(v')\leq \dg(u)+1$ then $(u,v')$ is not violated and the condition would be false. If $\dg(v')\geq \dg(u)+3$ then it must have been that $\dg(v')\geq \dg(u)+2$ before the deletion of edge $(u,v)$, contradicting the assumption that the input graph $G$ has an orientation for which Invariant~\ref{inv:degree_inequality} holds.
\end{proof}

\begin{lemma}\label{lem:valid_graph_deletion_recursive_call}
 During the execution of \rd$(G,(u,v))$ on a graph $G$ whose orientation satisfies Invariant~\ref{inv:degree_inequality}, each recursive call made to \rd \ has an input graph with an orientation that satisfies Invariant~\ref{inv:degree_inequality}.
\end{lemma}

\begin{proof}
The proof is by induction on the number of recursive steps of the algorithm. When edge $(u,v)$ is deleted and removed, any violated edge must be incoming into $u$. The algorithm finds such an edge $(v',u)$ if it exists by using $H_u$, and once it finds $(v',u)$, that edge is duplicated with a flipped orientation ($\orient u v'$), thereby reverting $\dg(u)$ to its original value prior to the deletion of $(u,v)$. All edges of the new multi-graph other than the new copy of $(u,v')$ are valid by induction, and the new copy of $(u,v')$ is valid as prior to the recursive call we had $\dg(v')= \dg(u)+1$ due to Lemma~\ref{lem:violated_edge_deletion_bound}.
\end{proof}

As a direct consequence of Lemma~\ref{lem:valid_graph_deletion_recursive_call} we have the following.

\begin{lemma}\label{cor:deletion_correct}
At the end of the execution of \rd \ on an input graph which has an orientation satisfying Invariant~\ref{inv:degree_inequality}, Invariant~\ref{inv:degree_inequality} holds for the resulting graph and orientation.
\end{lemma}

Next, we bound the total number of re-orientations performed and time spent during the deletion process.

\begin{lemma}\label{lem:deletion_main}
The total number of recursive calls (and hence re-orientations) of \rd \ due to a deletion of an edge from $G$ is at most $\Delta +1$, and the total runtime is bounded by $O({\Delta })$.
\end{lemma}

\begin{proof}
Consider an execution of \rd$(G,(u,v))$ and let $x$ be the out-degree of $u$ at the beginning of this execution. Furthermore, assume that during this execution we reach line 7 to call \rd$(G',(v',u))$ (notice that the difference between $G$ and $G'$ is two edges), and let $y$ be the out-degree of $v'$ at the start of this recursive call. Then due to Lemma~\ref{lem:violated_edge_deletion_bound} and the fact that we added a copy of $(u,v')$ in line 5 of \rd \ oriented from $u$ to $v'$, it must be that $y=x+1$. In other words, in each consecutive recursive call the out-degree of the first vertex variable is incremented, and so the maximum number of recursive steps is bounded by the maximum out-degree in the graph.

The running time of each recursive step is $O(1)$ due to Lemma~\ref{lem:structure}, yielding a total of $O(\Delta)$ time for the entire process. Finally, lines 9-11 are executed only once and take a total of $O(\Delta)$ time.

\end{proof}

\subsection{Conclusion}

\begin{theorem}\label{thm:efficient}
There exists a deterministic algorithm for maintaining an orientation of a fully dynamic
graph on $n$ vertices while supporting the following:
\begin{itemize}[noitemsep,nolistsep]
\item The maximum out-degree is $\Delta  \leq \inf_{\beta>1} \{\beta\cdot \alpha(G) + \log_\beta n\}$,
\item The worst-case time to execute an edge insertion is $O(\Delta^2)$,
\item The worst-case time to execute an edge deletion is $O(\Delta)$, and
\item The worst-case number of orientations performed per update is $\Delta+1$.
\end{itemize}
\end{theorem}
\noindent{\bf Remark:} The statement of Theorem~\ref{thm:efficient} is valid without any knowledge of an upper-bound $\alpha$ on $\alpha(G)$, where $\alpha(G)$ may change as the graph changes.
To see this, notice that the algorithms presented in this Section do not make any assumption on the arboricity (which stands in contrast to the more efficient algorithms that will be presented in Section~\ref{sec:efficient_algo}).

\section{A More Efficient Algorithm}\label{sec:efficient_algo}
In this section we present a more efficient algorithm that improves the insertion update time from $O(\Delta^2)$ to $O(\gamma \cdot \Delta)$, without increasing any of the other measures.

\subsection{An Intermediate Invariant}

So far we have introduced two invariants. On one extreme, the stronger Invariant~\ref{inv:degree_inequality} guarantees that
all edges are valid, and this led to our simple algorithm in Section~\ref{sec:algo}. On the other extreme,
the weaker Invariant~\ref{weakest} only guarantees that $\gamma$ outgoing edges of each vertex are valid. On an intuitive level, the benefit of having the weaker Invariant~\ref{weakest} being maintained comes into play during the insertion process of edge $(u,v)$ that is oriented as $\orient u v$, where instead of scanning all of the outgoing edges of $u$ looking for a violated edge, it is enough to scan only $\gamma$ edges. If such a guarantee could be made to work, the insertion update time would be reduced to $O(\gamma\ \cdot \Delta)$. However, it is unclear how to efficiently maintain Invariant~\ref{weakest} as deletions take place. Specifically, when one of the $\gamma$ outgoing valid edges of a vertex is deleted, it is possible that there is no other valid outgoing edge to replace it.

Our strategy is not to maintain Invariant~\ref{weakest} directly, but rather to define and maintain an intermediate invariant (see Invariant~\ref{inv:intermediate}), which is stronger than Invariant~\ref{weakest} but still weak enough so that we only need to scan $\gamma$ outgoing edges of $u$ during the insertion process. The additional strength of the intermediate invariant will assist us in efficiently supporting deletions. Before stating the invariant, we define the following. For any $i \ge 1$, an edge $(u,v)$ oriented as $\orient u v$ is called \emph{$i$-valid} if $\dg(v) \ge \dg(u) - i$; if it is not $i$-valid then it is \emph{$i$-violated}. We also say that a vertex $w$ is \emph{spectrum-valid} if the set $E_w$ of its outgoing edges can be partitioned into $q=q_w=\lceil\frac{|E_w|}{\gamma}\rceil$ sets $E_w^1,\cdots,E_w^{q}$ such that for each $1\leq i \leq q$, the following holds: (1) $|E_w^i|=\gamma$ (except for the residue set $E_w^{q}$ which contains the remaining $|E_w|-(q-1)\cdot \gamma$ edges, i.e., $|E_w^{q}| = |E_w|-(q-1)\cdot \gamma$), and (2) all edges in $E_w^i$ are $i$-valid. If a vertex is not spectrum-valid then it is \emph{spectrum-violated}.

\begin{inv} \label{inv:intermediate}
Each vertex $w$ is spectrum-valid.
\end{inv}

We will call $E_w^1$ ($E_w^{q}$) the first (last) set of edges for $w$. To give some intuition as to why Invariant~\ref{inv:intermediate} helps us support deletions efficiently, notice that once an edge $(u,v)$ that is oriented as $\orient u v$ is deleted and needs to be replaced, it will either be replaced by a flip of some violated incoming edge (which will become valid after the flip), or it can be replaced by one of the edges from $E_u^2$, as these edges were previously 2-valid, and after the deletion they are all 1-valid. We emphasize already here that during the insertion process we do not scan the $\gamma$
edges of the first set (i.e., those that are guaranteed to be 1-valid prior to the insertion),
but rather scan the $\gamma$ (in fact, $\gamma-1$) edges of the last set (and possibly of the set before last)
that are only guaranteed to be $q$-valid.
This somewhat counter-intuitive strategy is described and explained in detail in Section~\ref{subsec:eff_insertion}.

In order to facilitate the use of Invariant~\ref{inv:intermediate}, each vertex $w$ will maintain its outgoing edges in a doubly linked list $\mathcal{L}_w$. We say that $\mathcal{L}_w$ is \emph{valid} if for every $1\leq i \leq q$, the edges between location $\gamma\cdot (i-1)+1$ and location $(\gamma\cdot i)$ in the list are all $i$-valid. These locations for a given $i$ are called the $i$-block of $\mathcal{L}_w$. So, in a valid $\mathcal{L}_w$ the first location must be $1$-valid and belongs to the $1$-block, the last location must be $q$-valid and belongs to the $q$-block, etc. Note that for $i=q$
the number of locations (i.e., $|E_w|-(q-1)\cdot \gamma$) may be smaller than $\gamma$.
If $\mathcal{L}_w$ is not valid then it is \emph{violated}.

\subsection{Insertions}\label{subsec:eff_insertion}

Suppose that edge $(u,v)$ is added to $G_{i-1}$ thereby obtaining $G_i$. The process of inserting the new edge is performed as in Section~\ref{sec:algo} with the following modifications. Instead of scanning all outgoing edges of $u$ in order to find a violated edge, we only scan the \emph{last} $\gamma-1$ edges in $\mathcal{L}_u$; if there are less than $\gamma-1$ edges then we scan them all. If one of these edges, say $(u,v')$, is violated then we remove $(u,v')$ from the graph, replace $(u,v')$ with $(u,v)$ in $\mathcal{L}_u$, and recursively insert $(u,v')$ with the flipped orientation (just like in Section~\ref{sec:algo}). If all of these edges are valid, we move them together with the new edge $(u,v)$ to the front of $\mathcal{L}_u$.

The full details of the insertion process appear in Algorithm~\ref{alg3}.

\begin{algorithm}
\caption{\eri$(G,(u,v))$}
/* Assume without loss of generality that $\dg(u) \leq \dg(v) $ */
\begin{algorithmic}[1]
\label{alg3}
\STATE add $(u,v)$ to $G$ with orientation $\orient u v$
\STATE {\sf Insert}$(H_{v},u,\dg(u)-1)$  /* this key will be incremented in line 12 if needed */
\STATE Let $S_u$ be the last $min(\dg(u),\gamma-1)$ edges in $\mathcal{L}_u$.
\FOR {$v'\in S_u$}
\IF {$\dg(u) > \dg(v')+1$ }
\STATE remove $(u,v')$ from $G$ /* now edge $(u,v')$ is missing */
\STATE {\sf Delete}$(H_{v'}, u)$
\STATE replace $(u,v')$ with $(u,v)$ in $\mathcal{L}_u$
\STATE \eri$(G,(v',u))$ /* recursively insert edge $(u,v')$, but orient it as $\orient {v'} u$ */
\STATE return
\ENDIF
\ENDFOR
\FOR {$v'\in N^+(u)$}
\STATE {\sf Increment}$(H_{v'},u)$ \ENDFOR
\STATE {\sf IncrementCenter}$(H_{u})$
\STATE move $S_u$ to the front of $L_u$
\STATE add $(u,v)$ to the front of $L_u$

\end{algorithmic}
\end{algorithm}

\paragraph{Correctness and Running Time}

First, we remark that an appropriate variation of Observation~\ref{obs:violated_edge_insertion_bound} is immediately true here as well, but this time for \eri \  and line 5. What needs to be proven is an appropriate variation of Lemma~\ref{lem:valid_graph_insertion_recursive_call} which we describe next.

\begin{lemma}\label{lem:eff_valid_graph_insertion_recursive_call}
 During the execution of \eri$(G,(u,v))$ on a graph $G$ whose orientation satisfies Invariant~\ref{inv:intermediate} and for every vertex $w$ the list $\mathcal{L}_w$ is valid, each recursive call made to \eri \ has an input graph with an orientation that satisfies Invariant~\ref{inv:intermediate} and for every vertex $w$ the list $\mathcal{L}_w$ is valid.
\end{lemma}

\begin{proof}
The proof is by induction on the number of recursive steps of the algorithm. When edge $(u,v)$ is inserted and oriented
as $\orient u v$, the only vertex in the graph which may be spectrum-violated is $u$. The algorithm brute-force checks $\gamma-1$ outgoing edges of $u$, specifically, the $\gamma -1$ last edges in ${\cal L}_u$.
If a violated edge $(u,v')$ edge is found, that edge is removed thereby reverting $\dg(u)$ to its original value prior to the insertion of edge $(u,v)$. In addition, the new edge $(u,v)$ replaces $(u,v')$ in $\mathcal{L}_u$. What remains to be proven is that $\mathcal{L}_u$ is valid prior to the recursive call in line 9, as the rest of the lists are all valid by induction. To see that $\mathcal{L}_u$ is indeed valid, notice that the edge $(u,v')$ must be $1$-valid, and hence having it replace $(u,v')$ will keep $\mathcal{L}_u$ valid as well.
\end{proof}

\begin{lemma}\label{lemma:eff_insertion_correct}
At the end of the execution of \eri \ on an input graph which has an orientation satisfying Invariant~\ref{inv:degree_inequality} and where for every vertex $w$ the list $\mathcal{L}_w$ is valid, Invariant~\ref{inv:intermediate} holds for the resulting graph and orientation, and every vertex $w$ in the resulting graph has a valid $\mathcal{L}_w$.
\end{lemma}

\begin{proof}
The end of the execution of \eri \ can only happens once we reach line 11 (when calling \eri \ with edge $(u,v)$ such that $\dg(u) \le \dg(v)$). This means that we have checked the last $\gamma-1$ edges in $\mathcal{L}_u$ and they are all valid.
Let $d$ and $d'$ denote the value of $\dg(u)$ immediately before and after the insertion of $(u,v)$, respectively,
where $d' = d+1$.
Notice that if an edge $(u,w)$ that is oriented $\orient u w$ was $i$-valid before the insertion of $(u,v)$, for each $i$, it is possible that it is no longer $i$-valid after the insertion, but it must be $(i+1)$-valid. Also, due to Lemma~\ref{lem:eff_valid_graph_insertion_recursive_call}, the list $\mathcal{L}_u$ was valid at the start of this call to \eri. By moving the $\gamma-1$ examined edges from the end of $\mathcal{L}_u$ to its front, in addition to adding the new edge $(u,v)$ to the front of $\mathcal{L}_u$, we have pushed all of the other edges down a block (from the $i$-block to the $(i+1)$-block, for each $i$), and so $\mathcal{L}_u$ is valid after line 15. This immediately implies that the vertex $u$ is spectrum-valid after line 15 is executed. Furthermore, every vertex $w\neq u$ in the graph has a valid $\mathcal{L}_w$ by Lemma~\ref{lem:eff_valid_graph_insertion_recursive_call}, and so Invariant~\ref{inv:intermediate} holds.
\end{proof}

\begin{lemma}\label{lem:eff_insertion_main}
The total number of recursive calls (and hence re-orientations) of \eri \ due to an insertion into $G$ is at most $\Delta +1$, and the total time spent is bounded by $O(\gamma\cdot \Delta )$.
\end{lemma}
\begin{proof}
For the runtime, notice that just like in Lemma~\ref{lem:insertion_main}, the number of re-orientations during an insertion process is still bounded by $O(\Delta)$. But now, each recursive step costs $O(\gamma)$ time, except for the last step which costs $O(\Delta)$ time.
\end{proof}

\subsection{Deletions}\label{subsec:eff_deletion}

Suppose that edge $(u,v)$ is deleted from $G_{i-1}$ thereby obtaining $G_i$. The process of deleting the edge is performed as in Section~\ref{sec:algo} with the following modifications. If an edge incoming into $u$, say $(u,v')$, is violated and is flipped (just like in Section~\ref{sec:algo}), then we replace $(u,v)$ with $(u,v')$ in $\mathcal{L}_u$ and continue recursively to delete the original copy of $(u,v')$. If all incoming edges of $u$ are valid, we remove $(u,v)$ from $\mathcal{L}_u$.
The full details of the deletion process appear in Algorithm~\ref{alg4}.

\begin{algorithm}
\caption{\erd$(G,(u,v))$}
/* Assume without loss of generality that edge $(u,v)$ is oriented as $\orient u v$.
If there is another copy of $(u,v)$ oriented as $\orient v u$, assume without loss of generality that $\dg(u) \geq \dg(v) $ */
\begin{algorithmic}[1]
\label{alg4}
\STATE remove $(u,v)$ from $G$ /* if there are two copies of $(u,v)$, delete the one oriented as $\orient u v$ */
\STATE {\sf Delete}$(H_{v},u)$
\STATE $v' \leftarrow$ ReportMax$(H_{u})$
\IF {$\dg (v') > \dg(u)+1$}
\STATE add $(u,v')$ to $G$ with orientation $\orient u {v'}$ /* now there are two copies of $(u,v')$ */
\STATE replace $(u,v)$ with $(u,v')$ in $\mathcal{L}_v$
\STATE {\sf Insert}$(H_{v'}, u,\dg(u))$
\STATE \erd$(G,(v',u))$ /* recursively delete the copy of $(v',u)$ oriented as $\orient {v'} u$ */
\STATE return
\ENDIF
\FOR {$v'\in N^+(u)$}
\STATE {\sf Decrement}$(H_{v'},u)$
\ENDFOR
\STATE {\sf DecrementCenter}$(H_{u})$
\STATE Remove $(u,v)$ from $\mathcal{L}_v$
\end{algorithmic}
\end{algorithm}

\paragraph{Correctness and Running Time}

First, we remark that an appropriate variation of Lemma~\ref{lem:violated_edge_deletion_bound} is immediately true here as well, but this time for \erd. What remains to be proven is an appropriate variation of Lemma~\ref{lem:valid_graph_deletion_recursive_call} which we describe next.

\begin{lemma}\label{lem:eff_valid_graph_deletion_recursive_call}
 During the execution of \erd$(G,(u,v))$ on a graph $G$ whose orientation satisfies Invariant~\ref{inv:intermediate} and for every vertex $w$ the list $\mathcal{L}_w$ is valid, each recursive call made to \erd \ has an input graph with an orientation that satisfies Invariant~\ref{inv:intermediate} and for every vertex $w$ the list $\mathcal{L}_w$ is valid.
\end{lemma}

\begin{proof}
The proof is by induction on the number of recursive steps of the algorithm. When edge $(u,v)$ is deleted and removed, any spectrum-violated vertex must be an incoming neighbor of $u$. In such a case, there must exist an edge incoming into $u$ which is violated. The algorithm finds such an edge $(v',u)$ (if exists) using $H_u$, and once it finds $(v',u)$, that edge is duplicated with a flipped orientation ($\orient u v'$), thereby reverting $\dg(u)$ to its original value prior to the deletion of $(u,v)$. In addition, the new edge $(u,v')$ replaces the deleted edge $(u,v)$ in $\mathcal{L}_u$. Now, all the vertices of the new multi-graph except for $u$ are spectrum-valid and their lists are valid by induction. The list $\mathcal{L}_u$ is also valid as prior to the recursive call we had $\dg(v')= \dg(u)+1$ due to the appropriate variation of Lemma~\ref{lem:violated_edge_deletion_bound}, and so replacing $(u,v)$ with $(u,v')$ in $\mathcal{L}_u$ still maintains its validity.
\end{proof}

\begin{lemma}\label{lemma:eff_deletion_correct}
At the end of the execution of \erd \ on an input graph which has an orientation satisfying Invariant~\ref{inv:degree_inequality} and where for every vertex $w$ the list $\mathcal{L}_w$ is valid, Invariant~\ref{inv:intermediate} holds for the resulting graph and orientation, and every vertex $w$ in the resulting graph has a valid $\mathcal{L}_w$.
\end{lemma}

\begin{proof}
The end of the execution of \erd \ can only happens once we reach line 10 (when calling \erd \  with edge $(u,v)$). This means that all incoming edges of $u$ must be valid. Also, due to Lemma~\ref{lem:eff_valid_graph_deletion_recursive_call}, the list $\mathcal{L}_u$ was valid at the start of this call to \eri. If $(u,v)$ was in the $j$-block of $\mathcal L_u$ then for each $i>j$ as a result of removing $(u,v)$ from $\mathcal L_u$, there is one edge that moves from the $i$-block to the $(i-1)$ block. Notice that this movement does not take any additional time as the definition of blocks depends on locations only. Being that if an edge $(u,w)$ oriented $\orient u w$ was $i$-valid before the removal of $(u,v)$ it is necessarily $(i-1)$-valid after the removal (as the degree of $u$ was decremented), the movement of edges between blocks guarantees that $\mathcal L_u$ is valid after the deletion.
\end{proof}

\begin{lemma}\label{lem:eff_deletion_main}
The total number of recursive calls (and hence re-orientations) of \erd \ due to a deletion of an edge from $G$ is at most $\Delta +1$, and the total time spent is bounded by $O(\Delta )$.
\end{lemma}
\begin{proof}
For the runtime, notice that just like in Lemma~\ref{lem:deletion_main}, the number of re-orientations during an insertion process is still bounded by $O(\Delta)$. Each recursive step costs $O(1)$ time, except for the last step which costs $O(\Delta)$ time.
\end{proof}

\subsection{Conclusion}

\begin{theorem}\label{thm:simple}
There exists a deterministic algorithm for maintaining an orientation of a fully dynamic
graph on $n$ vertices that has arboricity at most $\alpha$ (at all times),
while supporting the following:
\begin{itemize}[noitemsep,nolistsep]
\item The maximum out-degree is $\Delta  \leq \inf_{\beta>1} \{\beta\cdot \alpha(G) + \log_\beta n\}$,
\item The worst-case time to execute an edge insertion is $O(\alpha \cdot \beta \cdot \Delta)$,
\item The worst-case time to execute an edge deletion is $O(\Delta)$, and
\item The worst-case number of orientations performed per update is $\Delta+1$.
\end{itemize}
\end{theorem}

\section*{Acknowledgments}
The fourth-named author is grateful to Ofer Neiman for helpful discussions.

{\small
\bibliographystyle{alphaurlinit}
\bibliography{tsvi}

\begin{thebibliography}{ABGS12}

\bibitem[ABGS12]{ABGS12}
A.~Anand, S.~Baswana, M.~Gupta, and S.~Sen.
\newblock Maintaining approximate maximum weighted matching in fully dynamic
  graphs.
\newblock In {\em IARCS Annual Conference on Foundations of Software Technology
  and Theoretical Computer Science, FSTTCS}, pages 257--266, 2012.
\newblock \href {http://dx.doi.org/10.4230/LIPIcs.FSTTCS.2012.257}
  {\path{doi:10.4230/LIPIcs.FSTTCS.2012.257}}.

\bibitem[AMZ97]{AMZ97}
S.~R. Arikati, A.~Maheshwari, and C.~D. Zaroliagis.
\newblock Efficient computation of implicit representations of sparse graphs.
\newblock {\em Discrete Applied Mathematics}, 78(1-3):1--16, 1997.
\newblock \href {http://dx.doi.org/10.1016/S0166-218X(97)00007-3}
  {\path{doi:10.1016/S0166-218X(97)00007-3}}.

\bibitem[AYZ95]{AYZ95}
N.~Alon, R.~Yuster, and U.~Zwick.
\newblock Color-coding.
\newblock {\em J. ACM}, 42(4):844--856, 1995.
\newblock \href {http://dx.doi.org/10.1145/210332.210337}
  {\path{doi:10.1145/210332.210337}}.

\bibitem[BF99]{BF99}
G.~S. Brodal and R.~Fagerberg.
\newblock Dynamic representation of sparse graphs.
\newblock In {\em Algorithms and Data Structures, 6th International Workshop,
  WADS}, pages 342--351, 1999.
\newblock \href {http://dx.doi.org/10.1007/3-540-48447-7_34}
  {\path{doi:10.1007/3-540-48447-7_34}}.

\bibitem[CE91]{CE91}
M.~Chrobak and D.~Eppstein.
\newblock Planar orientations with low out-degree and compaction of adjacency
  matrices.
\newblock {\em Theor. Comput. Sci.}, 86(2):243--266, 1991.
\newblock \href {http://dx.doi.org/10.1016/0304-3975(91)90020-3}
  {\path{doi:10.1016/0304-3975(91)90020-3}}.

\bibitem[CSW07]{CSW07}
J.~A. Cain, P.~Sanders, and N.~Wormald.
\newblock The random graph threshold for $k$-orientiability and a fast
  algorithm for optimal multiple-choice allocation.
\newblock In {\em 18th Annual ACM-SIAM Symposium on Discrete Algorithms}, pages
  469--476. SIAM, 2007.
\newblock Available from:
  \url{http://dl.acm.org/citation.cfm?id=1283383.1283433}.

\bibitem[DT13]{DT13}
Z.~Dvorak and V.~Tuma.
\newblock A dynamic data structure for counting subgraphs in sparse graphs.
\newblock In {\em Algorithms and Data Structures - 13th International
  Symposium, WADS}, pages 304--315, 2013.
\newblock \href {http://dx.doi.org/10.1007/978-3-642-40104-6_27}
  {\path{doi:10.1007/978-3-642-40104-6_27}}.

\bibitem[EKM12]{EKM12}
D.~Eisenstat, P.~N. Klein, and C.~Mathieu.
\newblock An efficient polynomial-time approximation scheme for steiner forest
  in planar graphs.
\newblock In {\em Proceedings of the Twenty-Third Annual ACM-SIAM Symposium on
  Discrete Algorithms, SODA}, pages 626--638, 2012.
\newblock Available from:
  \url{http://dl.acm.org/citation.cfm?id=2095116.2095169}.

\bibitem[Epp09]{Eppstein09}
D.~Eppstein.
\newblock All maximal independent sets and dynamic dominance for sparse graphs.
\newblock {\em ACM Transactions on Algorithms}, 5(4), 2009.
\newblock \href {http://dx.doi.org/10.1145/1597036.1597042}
  {\path{doi:10.1145/1597036.1597042}}.

\bibitem[Eri06]{Erickson06}
J.~Erickson.
\newblock
  \url{http://www.cs.uiuc.edu/~jeffe/teaching/datastructures/2006/problems/Bill-arboricity.pdf},
  2006.
\newblock Retrieved November 2013.

\bibitem[GP13]{GP13}
M.~Gupta and R.~Peng.
\newblock Fully dynamic $(1+\epsilon)$-approximate matchings.
\newblock In {\em Proceedings of the 54th Anual IEEE Symposium on Foundations
  of Computer Science, FOCS (to appear)}, 2013.

\bibitem[GW92]{GW92}
H.~N. Gabow and H.~H. Westermann.
\newblock Forests, frames, and games: Algorithms for matroid sums and
  applications.
\newblock {\em Algorithmica}, 7(5{\&}6):465--497, 1992.
\newblock \href {http://dx.doi.org/10.1007/BF01758774}
  {\path{doi:10.1007/BF01758774}}.

\bibitem[IL93]{IL93}
Z.~Ivkovic and E.~L. Lloyd.
\newblock Fully dynamic maintenance of vertex cover.
\newblock In {\em Graph-Theoretic Concepts in Computer Science, 19th
  International Workshop, WG}, pages 99--111, 1993.

\bibitem[KK06]{KK06}
L.~Kowalik and M.~Kurowski.
\newblock Oracles for bounded-length shortest paths in planar graphs.
\newblock {\em ACM Transactions on Algorithms}, 2(3):335--363, 2006.
\newblock \href {http://dx.doi.org/10.1145/1159892.1159895}
  {\path{doi:10.1145/1159892.1159895}}.

\bibitem[Kow07]{Kowalik07}
L.~Kowalik.
\newblock Adjacency queries in dynamic sparse graphs.
\newblock {\em Inf. Process. Lett.}, 102(5):191--195, 2007.
\newblock \href {http://dx.doi.org/10.1016/j.ipl.2006.12.006}
  {\path{doi:10.1016/j.ipl.2006.12.006}}.

\bibitem[NS13]{NS13}
O.~Neiman and S.~Solomon.
\newblock Simple deterministic algorithms for fully dynamic maximal matching.
\newblock In {\em Proceedings of the 45th ACM Symposium on Theory of Computing,
  STOC}, pages 745--754, 2013.
\newblock \href {http://dx.doi.org/10.1145/2488608.2488703}
  {\path{doi:10.1145/2488608.2488703}}.

\bibitem[NW61]{NASH61}
C.~S. J.~A. Nash-Williams.
\newblock Edge-disjoint spanning trees in finite graphs.
\newblock {\em Journal of the London Mathematical Society}, 36(1):445--450,
  1961.
\newblock \href {http://dx.doi.org/10.1112/jlms/s1-36.1.445}
  {\path{doi:10.1112/jlms/s1-36.1.445}}.

\bibitem[NW64]{Nash64}
C.~S. J.~A. Nash-Williams.
\newblock Decomposition of finite graphs into forests.
\newblock {\em Journal of the London Mathematical Society}, 39(1):12, 1964.
\newblock \href {http://dx.doi.org/10.1112/jlms/s1-39.1.12}
  {\path{doi:10.1112/jlms/s1-39.1.12}}.

\bibitem[OR10]{OR10}
K.~Onak and R.~Rubinfeld.
\newblock Maintaining a large matching and a small vertex cover.
\newblock In {\em Proceedings of the 42nd ACM Symposium on Theory of
  Computing}, pages 457--464, 2010.
\newblock \href {http://dx.doi.org/10.1145/1806689.1806753}
  {\path{doi:10.1145/1806689.1806753}}.

\end{thebibliography}
}

\appendix

\section{Selected Applications}\label{app:applications}

\paragraph{Maintaining a maximal matching in a fully dynamic graph.}
As part of the growing interest in algorithms for dynamically-changing graphs, a series of recent results show how to maintain a maximal or approximately maximum matching in a fully dynamic graph \cite{IL93,OR10,ABGS12,NS13,GP13}.
In particular, Neiman and Solomon~\cite{NS13} showed recently that one can maintain a maximal matching in $O(\frac{\log n}{\log\log n})$ {\em amortized} time per update for graphs with constant arboricity, using the algorithm of Brodal and Fagerberg~\cite{BF99}. By replacing Brodal and Fagerberg's algorithm with ours, we immediately achieve $O(\log n)$ {\em worst-case} time per insertion or deletion of an edge.

\begin{theorem}
Let $G=(V,E)$ be an undirected fully dynamic graph. Then one can deterministically maintain a maximal matching of $G$ such that the worst-case time per edge insertion or deletion is $O(\alpha(G)\cdot \beta\cdot \Delta )$ and $O(\Delta )$, respectively,
where $\Delta  = \inf_{\beta>1} \{\beta\cdot \alpha(G) + \log_\beta n\}$.
\end{theorem}

\paragraph{Adjacency queries.}
Perhaps the most fundamental application for finding a $c$-orientation
of an undirected graph with bounded arboricity is to obtain a (deterministic) data structure that answers adjacency queries quickly .
As mentioned above, once we have a $\Delta$-orientation, we can answer adjacency queries in $O(\Delta)$ time.
This can be improved further if the graph vertices have names (identifiers)
that are distinct and comparable, as using ideas from Kowalik~\cite{Kowalik07} we can store the outgoing edges of every vertex dynamic deterministic dictionary (see~\cite{}), ordered by the names of the vertices. Using this dictionary, we can answer adjacency queries in $O(\log \log \Delta)$ time but the update times suffer an overhead multiplicative factor of $O(\log \log \Delta)$.
For example, as long as $\alpha \leq (\log n)^{O(1)}$,
adjacency queries can be answered in $O(\log \log \log n)$ worst-case time.
Thus, our bounds immediately yield the following result.

\begin{theorem}
Let $G=(V,E)$ be an undirected fully dynamic graph. Then one can answer deterministic adjacency queries on $G$ in $O(\log \log \Delta )$ worst-case time where the deterministic worst-case time per edge insertion or deletion is $O(\alpha(G)\cdot\beta\cdot\Delta \cdot \log \log \Delta)$ and $O(\Delta \cdot \log \log \Delta)$ respectively,
where $\Delta  = \inf_{\beta>1} \{\beta\cdot \alpha(G) + \log_\beta n\}$.
\end{theorem}

\paragraph{Dynamic matrix by vector multiplication.}
Consider an everywhere-sparse matrix $A=\{a_{ij}\}$ of size $n \times n$ in the sense that every submatrix of $A$ is also sparse. For simplicity assume that $A$ is symmetric.
Suppose we have also an $n$-dimensional vector $\vec{x}$,
and we want to support coordinate-queries to $\vec{y}=A\cdot \vec{x}$,
namely, given a query index $i$, we should report $y_i=\sum_{j=1}^n a_{ij} x_j$.
Now suppose that $A$ and $\vec{x}$ keep changing,
and we want to support such changes (each change updates one entry) quickly,
and still be able to quickly answer coordinate-queries to $\vec{y}$.

To solve this, we can treat the index set $[n]$ as vertices in a graph,
whose adjacency matrix is $A$.
For each vertex $i$, we maintain a weighted summation of the $x$ value's of its incoming neighbors, where the weights are derived from $A$.
When a value in $A$ changes, we write the new weight of the corresponding edge
at the vertex that this edge is oriented into.
When a coordinate of $\vec{x}$ changes,
we updates for the corresponding vertex all its the outgoing neighbors.
Now when a query for $y_ii$ is performed, the corresponding vertex $i$
already has the data from its incoming neighbors, and only needs to gather
data from his outgoing neighbors.
Since every vertex has at most $c$ outgoing neighbors,
we thus establish the following result.

\begin{theorem}
Let $A_{n\times n}$ be a symmetric matrix, and let $G$ be the undirected graph whose adjacency matrix is $A$,
and let $\Delta = \inf_{\beta>1} \{\beta\cdot \alpha(G) + \log_\beta n\}$.
Let $\vec{x}$ be a vector of dimension $n$. Then we can support changes to $A$ in $O(\alpha \cdot \beta \cdot \Delta)$ worst-case time, changes to $\vec{x}$ in $O(\Delta)$ worst-case time, and for coordinate-query $i\in[n]$ we can report $y_i=\sum_{j=1}^n a_{ij} x_j$ in $O(\Delta)$ worst-case time.
This data structure uses space that is linear in $n$ and in the number of non-zeros in $A$.
\end{theorem}

\paragraph{Dynamic shortest-path queries in planar graphs.}
In a \emph{shortest-path of length at most $k$ query} on vertices $u$ and $v$, we wish to find the shortest path between $u$ and $v$ with total length (number of edges) at most $k$. Kowalik and Kurowski in~\cite{KK06} showed that for constant $k$ one can maintain a fully dynamic planar graph such that answering shortest-path of length at most $k$ queries takes $O(1)$ time, and updates take amortized polylogarithmic time. As part of their algorithm, Kowalik and Kurowski make use of the algorithms of Brodal and Fagerberg~\cite{BF99}, and so replacing those algorithms with ours (together with some straightforward machinery) we obtain the following theorem, with worst-case bounds.

\begin{theorem}
An undirected planar graph $G=(V,E)$ can be maintained through fully dynamic updates such that queries for shortest-path of length at most $k$ can be answered in worst-case $O(\log^{O(k)} n)$ time, and the worst-case update time for both insertions and deletions is $O(\log^{O(k)} n)$.
\end{theorem}

\end{document}